\documentclass[11pt,letterpaper]{article}
\usepackage{epic, eepic}
\usepackage{units}
\usepackage{hyperref}
\usepackage{amsthm, amsmath, amssymb,amscd}
\usepackage[all]{xy}
\usepackage{mathrsfs}
\usepackage{graphicx}
\usepackage{epstopdf}
\usepackage{tikz}
\usepackage{enumerate}

\newtheorem{thm}{Theorem}[subsection]

\newtheorem{claim}{Claim}

\newtheorem{example}[thm]{Example}

\usepackage{fullpage}

\begin{document}
\title{Coverage statistics for sequence census methods} 
\author{Steven N. Evans, Valerie Hower and Lior Pachter}
\date{\today}

\maketitle

\begin{abstract}
\noindent
{\em Background:}
We study the statistical properties of fragment coverage in genome sequencing experiments. In an extension of the classic Lander-Waterman model, 
we consider the effect of the length distribution of fragments. We also introduce the notion of the {\em shape} of a coverage function, which can be used to detect abberations in coverage. The probability theory underlying these problems is essential for constructing models of current high-throughput sequencing experiments, where both sample preparation protocols and sequencing technology particulars can affect fragment length distributions. 
\\ \\
{\em Results:}
We show that regardless of fragment length distribution and under the mild assumption that fragment start sites are Poisson distributed, the fragments produced in a sequencing experiment can be viewed as resulting from a two-dimensional spatial Poisson process. We then study the jump skeleton of the the coverage function, and show that the induced trees are Galton-Watson trees whose parameters can be computed. 
\\ \\
{\em Conclusions:}
Our results extend standard analyses of shotgun sequencing that focus on coverage statistics at individual sites,
and provide a null model for detecting deviations from random
coverage in high-throughput sequence census based experiments. By focusing on fragments, we are also led to a new approach for visualizing sequencing data that should be of independent interest.

\end{abstract}
\section{Introduction}

The classic ``Lander-Waterman model'' \cite{Lander1988}  provides
statistical estimates for the read coverage in a whole genome shotgun
(WGS) sequencing experiment via the Poisson approximation to the
Binomial distribution. Although originally intended for estimating the
extent of coverage when mapping by fingerprinting random clones, the
Lander-Waterman model has served as an essential tool for estimating
sequencing requirements for modern WGS experiments \cite{Myers1997}. Although it makes a number of simplifying assumptions
(e.g. fixed fragment length and uniform fragment selection ) that are
violated in actual experiments, extensions and generalizations \cite{Wendl2005,Wendl2006} have continued to be developed and applied in a variety of settings.

The advent of ``high-throughput sequencing'', which refers to
massively parallel sequencing technologies has greatly increased the
scope and applicability of sequencing experiments. With the increasing
scope of experiments, new statistical questions about coverage statistics
have emerged. In particular, in the context of {\em sequence census
  methods}, it has become important to understand the
{\em shape} of coverage functions, rather than just coverage statistics
at individual sites.

Sequence census methods \cite{Wold} are experiments designed to assess
the content of a mixture of molecules via the creation of DNA
fragments whose abundances can be used to infer those of the original
molecules. The DNA fragments are identified by sequencing, and
the desired abundances inferred by solution of an inverse
problem. An example of a sequence census method is ChIP-Seq. In this
experiment, the goal is to determine the locations in the genome where
a specific protein binds. An antibody to the protein is used to ``pull
down'' fragments of DNA that are bound via a process called chromatin
immunoprecipitation (abbreviated by ChIP). These fragments form the
``mixture of molecules'' and after purifying the DNA, the fragments
are determined by sequencing. The resulting sequences are compared to
the genome, leading to a {\em coverage function}  that records, at each
site, the number of sequenced fragments that contained it. As with
many sequence census methods, ``noise'' in the experiment leads to
random sequenced fragments that may not correspond to bound DNA, and
therefore it is necessary to identify regions of the coverage function
that deviate from what is expected according to a suitable null
model.  

The purpose of this paper is not to develop methods for the analysis
of ChIP-Seq (or any other sequence census method), but rather to
present a null model for the shape of a coverage function that is of
general utility. That is, we propose a definition for the shape
of a fragment coverage function, and describe a random instance assuming that fragments are selected at
random from a genome, with lengths of fragments given by a known distribution. The distinction between our work and previous statistical studies of sequencing experiments, is that we go beyond the description of coverage at a single location, to a description of the change in coverage along a genome.

\section{The shape of a fragment coverage function}\label{treesection}

We begin by explaining what we mean by a {\em coverage function}. Given a
genome modeled as a string of fixed length $N$, a coverage function is a function
$f:\{1,\ldots,N\} \longrightarrow \mathbb{Z}_{\ge 0}$. The interpretation
of this function, is that $f(i)$ is the number of sequenced fragments
obtained from a sequencing experiment that cover position $i$ in the
genome. It is important to note that $N$ is typically large; for
example, the human genome consists of approximately $2.8$ billion
bases. Because $N$ is very large, we replace the finite set
$\{1,\ldots,N\}$ with $\mathbb{R}$, and re-define a coverage function
to be a function $f:\mathbb{R} \longrightarrow \mathbb{Z}_{\ge
  0}$. This helps to simplify our analysis.

We next introduce an object that describes a sequence coverage function's shape.  Our approach is motivated by recent applications of topology including persistent homology \cite{Carlsson-2009,Zomorodian-2005} and the use of critical points in shape analysis \cite{Biasotti-2008, deBerg-1997,Edelsbrunner-2003}.  For a given coverage function $f:\mathbb{R} \longrightarrow \mathbb{Z}_{\ge 0}$, we will define a rooted tree, which is a particular type of directed graph with all the directed edges pointing away from the root.  This tree $T_f$ is based on the \emph{upper-excursion sets of} $f$: $U_h:=\{(x,f(x)) | f(x) \ge h\}$, $h\in \mathbb{Z}_{\ge 0}$ and keeps track of how the sets $U_h$ evolve as $h$ decreases.  Long paths in $T_f$ represent features of the coverage function that persist through many values of $h$. 

Specifically, for each $h \in \mathbb{Z}_{\ge 0}$, let $C_h$ denote the set of connected components of the upper-excursion set  $U_h$.  We define the rooted tree $T_f=(V,E)$ as follows
\begin{itemize} 
\item Vertices in $V$ correspond to the connected components in the collection $\{ C_h\}_{h\in \mathbb{Z}_{\ge0}}$
\item $(i,j)\in E$ provided their corresponding connected components $c_i \in C_{h_i}$ and $c_j \in C_{h_j}$ with $h_i<h_j$ satisfy $h_i=h_j-1$ and $c_j \subset c_i$.
\end{itemize}
Note that the root of $T_f$ corresponds to the single connected component in $C_0$.  The tree $T_f$ is very similar to a contour tree \cite[\textsection 4.1]{Biasotti-2008}, which is built using level sets of a function, and a join tree \cite{Carr-2003}.  Indeed, suppose we ignore every vertex that is adjacent to only one vertex with greater height.  Then, the remaining vertices of $T_f$ correspond to (equivalence classes of) local extrema of $f$.  Each local maximum of $f$ yields the birth of a new connected component as we sweep down through $h\in \mathbb{Z}_{\ge 0}$ while a local minimum of $f$ merges connected components.  Since we do not require $f$ to have distinct critical values (as is frequently assumed), the vertices in $T_f$ can have arbitrary degrees, as is depicted in Figure \ref{curve-path-tree}C.  

In the sequel, we will use the following equivalent characterization that can be found in \cite[\textsection 2.3]{Evans-book}.  Given a coverage function $f:\mathbb{R}\longrightarrow \mathbb{Z}_{\ge 0}$ with $f(a)=f(b)=0$ and $f(x)>0$ for $x\in(a,b)$, we form an integer-valued sequence $x_0, \ldots, x_{2n}$ that records the changes in height of $f$ on the interval $[a,b]$.  The sequence $x_0, \ldots, x_{2n}$ consists of the $y$ values that $f$ travels through from $x_0:=f(a)=0$ to $x_{2n}:=f(b)=0$ and satisfies
\begin{eqnarray*}
 &x_0 = x_{2n} = 0, \\
 &x_i > 0 \mbox{ for } 0 < i < 2n, \\
&|x_i - x_{i-1}| = 1 \mbox{ for } 1 \le i \le 2n.
\end{eqnarray*}
Such a sequence is called a \emph{lattice path excursion away from} $0$.  Next, we define an equivalence relation on the set $\{0,1,\ldots,2n\}$ by setting
\[
i \equiv j \Longleftrightarrow x_i = x_j = \min_{i \le k \le j} x_k.
\]
The equivalence classes under this relation are in $1:1$ correspondence with the connected components in the upper-excursion sets of $\left. f\right|_{[a,b]}$.  One equivalence class is $\{0,2n\}$, and if $\{i_1, \ldots, i_p\}$ is an equivalence class with
$0<i_1 < i_2 < \ldots < i_p$ then $x_{i_1 - 1} = x_{i_1} - 1,$
whereas $x_{i_q - 1} = x_{i_q} + 1$ for $2 \le q \le p$.  Conversely,
any index $i$ with $x_{i - 1} = x_i - 1$ is the minimal element of
an equivalence class.   We use the minimal element of each equivalence class as its representative.  Thus, we can view the vertices of $T_{\left.f\right|_{[a,b]}}$ as the set $\{0\} \cup \{i | x_{i - 1} = x_i - 1\}.$  Two indices $i_1<i_2$ are adjacent in $T_{\left.f\right|_{[a,b]}}$ provided $x_{i_2} = x_{i_1} + 1$ and $x_k \ge x_{i_1}$ for $i_1 \le k \le i_2$.  Figure \ref{curve-path-tree} gives an example of a coverage function together with its lattice path excursion $(0,1,2,3,4,3,2,3,4,5,4,3,2,3,2,1,0)$ and rooted tree.  The minimal elements of each equivalence class in Figure \ref{curve-path-tree}B are depicted with red squares.  
\begin{figure}[htbp]
\begin{center}
\includegraphics[width=6.5in]{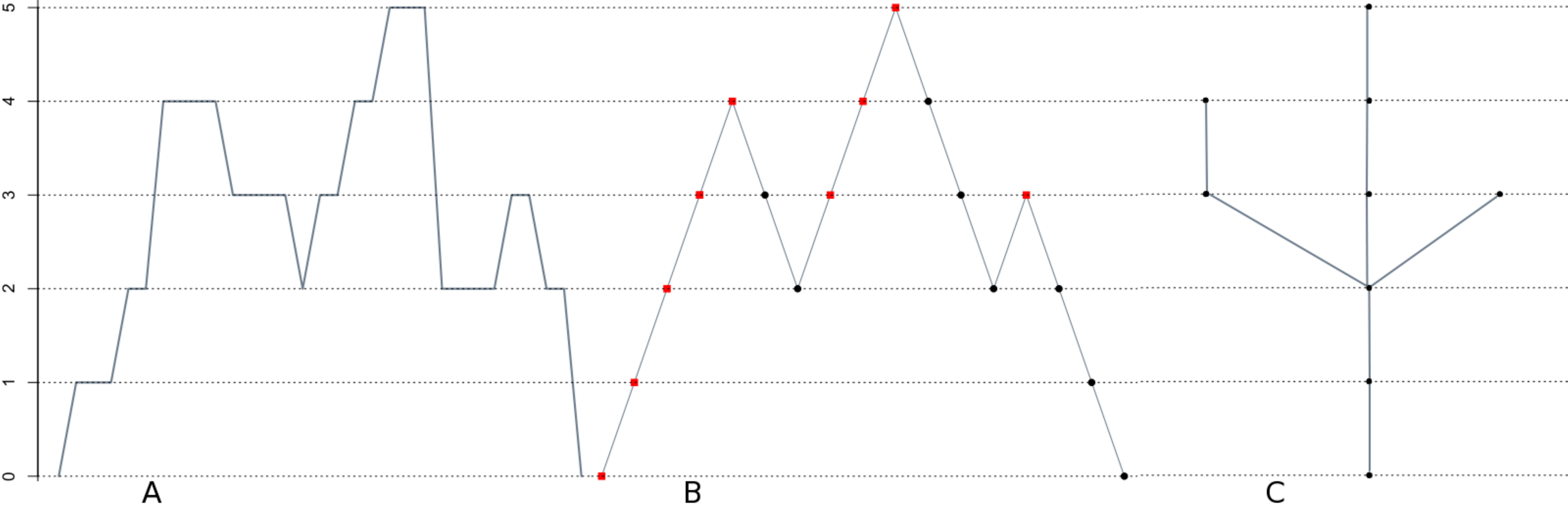}
\caption{ A coverage function (A) with its lattice path excursion (B) and rooted tree (C).}
\label{curve-path-tree}
\end{center}
\end{figure}

\section{Planar Poisson processes from sequencing experiments}
In order to model random coverage along the genome, we use a Poisson process to give random starting locations to the fragments.  Specifically, suppose that we have a stationary Poisson point process on $\mathbb{R}$ with intensity $\rho$.   At each point of the Poisson point process we lay down an interval that has that 
point as its left end-point. The lengths of the successive intervals are independent 
and identically distributed with common distribution $\mu$.  We will
use the notation $X$ for a coverage function built from this process
and $X_t$ for the height at a point $t$.      

Let $t_1, t_2, \cdots $ be the left-end points and $l_1,l_2,\cdots $ be the corresponding lengths of intervals.  The interval given by $(t_i,l_i)$ will cover a nucleotide $t_0$ provided $t_i\le t_0$ and $t_i+l_i \ge t_0$.  We can view this pictorially by plotting points $\{(t_j,l_j)\}$ in the plane.  Then $X_{t_0}$---the number of intervals covering $t_0$---is the number of points in the triangular region below.
\begin{figure}[htbp]
\begin{center}
\begin{tikzpicture}
\draw (-2,0) -- (5,0); 
\draw (2,0) -- (2,4); 
\draw (2,0) -- (-2,4)
node[midway, sloped,below] {$l=t_0-t$}; 
\fill[green!20!white] (2,0) -- (2,4) -- (-2,4) -- cycle;
\draw ( 2 cm,2pt) -- (2 cm,-3pt) node[anchor=north] {$t_0$}; 
\draw (4,2.5 )  node[anchor=north] {\small $(t,l)$-plane}; 

\end{tikzpicture}
\caption{A two dimensional view of a sequencing experiment.}
\label{wedge}
\end{center}
\end{figure}
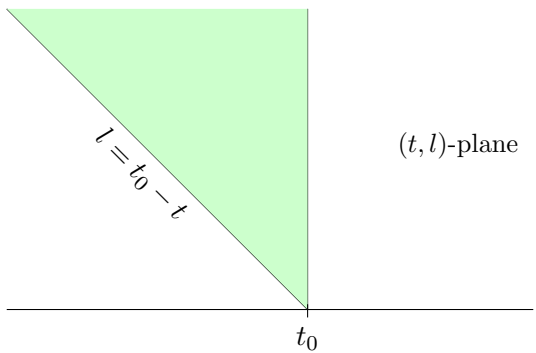
We now recall the definition of a two-dimensional Poisson process and refer the reader to \cite[\textsection 6.13]{Grimmett-book} or \cite[\textsection 2.4]{Daley-book} for the details.
Suppose $\Gamma$ is a locally finite measure on the Borel $\sigma$-algebra $\mathscr{B}(\mathbb{R}^2)$.  A random countable subset $\Pi$ of $\mathbb{R}^2$ is called a \emph{non-homogeneous Poisson process with mean measure} $\Gamma$ if, for all Borel subsets $A$, the random variables $N(A):=\#(A\cap \Pi)$ satisfy:
\begin{enumerate}
\item $N(A)$ has the Poisson distribution with parameter $\Gamma(A)$, and
\item If $A_1, \cdots , A_k$ are disjoint Borel subsets of $\mathbb{R}^2$, then $N(A_1),\cdots , N(A_k)$ are independent random variables.
\end{enumerate}
The following theorem is a consequence of  \cite[Proposition 12.3]{MR1876169}.
\begin{thm}\label{PoissonThm}
The collection $\{(t_i,l_i)\}$ of points obtained as described above is a non-homogeneous Poisson process with mean measure $\rho \, m \otimes \mu$.  Here $m$ is Lebesgue measure on $\mathbb{R}$.
\end{thm}
We compute the expected value $\mathbb{E}[X_t]=\rho \, m \otimes \mu(\mbox{wedge}):$
\begin{eqnarray*}
\rho \, m \otimes \mu(\mbox{wedge})&&=\rho\int_{-\infty}^t \int_{t-u}^{\infty}\mu(dv)du\\
&&=\rho\int_{-\infty}^t \mu((t-u,\infty))du \\
&&=\rho\int_0^{\infty}\mu((s,\infty))ds.
\end{eqnarray*}
\subsection{Fragment lengths have the exponential distribution}
We treat the simplest case first, namely the case where the distribution $\mu$ of fragment lengths is exponential with rate $\lambda$.  Then, we have $\mu((s,\infty))=\mathbb{P}\{l>s\}=e^{-\lambda s}$, and $$\mathbb{E}(X_t)=\rho\int_0^{\infty}e^{-\lambda s}ds=\frac{\rho}{\lambda}.$$  
\begin{claim}
The process $X$ is a stationary, time-homogeneous Markov process.
\end{claim}
\begin{proof} It is clear that $X$ is stationary because of the manner in which it is constructed from a Poisson process on $\mathbb{R}^2$ that has a distribution which is invariant under translations in the $t$ direction; that is, the random set $\{(t_i, l_i)\}$ has the same distribution as  $\{(t_i + t, l_i)\}$ for any fixed $t \in \mathbb{R}$.  Since $\mu$ is exponential, it is memoryless, meaning for any interval length $l$ with an exponential distribution
$$\mathbb{P}\{l>a+b |  l> a\}=\mathbb{P}\{l>b\}.$$
This means that probability that an interval covers $t_2$ knowing that it covers $t_1$ is the same as the probability that an interval starting at $t_1$ covers $t_2$.  Thus, the probability that $X_{t_2}=k$  given $X_t$ for $t \le t_1$ only depends on the value of $X_{t_1}$.  Indeed, in terms of time, $\mathbb{P}\{X_{t_2}=k|X_{t_1}= k^{\prime}\}$ depends only on $t_2-t_1$. 
\end{proof}

More specifically, X is a birth-and-death process with birth rate $\beta(k) = \rho$ in all states $k$ and death rate $\delta(k) = k\lambda $ in state $k \ge 1$.  Note that as the exponential distribution is the only distribution with the memoryless property, we lose the Markov property when $\mu$ is not exponential.  

To build the tree of \textsection \ref{treesection}, we are interested in the jumps of the coverage function $f(t)=X_t$.  We hence consider the jump chain of $X$---
a discrete-time Markov chain with transition matrix
\[
P(i,j) 
=
\begin{cases}
1,& \quad \text{if $i=0$ and $j=1$}, \\
\frac{\rho}{\rho + i\lambda },& \quad \text{if $i \ge 1$ and $j=i+1$}, \\
\frac{i \lambda }{\rho + i \lambda },& \quad \text{if $i \ge 1$ and $j=i-1$}, \\
0,& \quad \text{otherwise}.
\end{cases}
\]
Suppose now we have a lattice path excursion starting at $0$.  Given a vertex $v$ of the associated tree at height $k$, we are interested in the number of offspring (at height $k+1$) of this vertex.  Suppose $i_0$ is the minimal equivalence class representative for vertex $v$, and suppose $[i_0]=\{i_0, i_1, \cdots , i_n\}$ with $i_0<i_1<\cdots <i_n$.  Then, we have $x_{i_r}=k$ for $0\leq r \leq n$, $x_{i_r+1}= k+1$ for $0\leq r \leq n-1$, $x_{i_n+1}=k-1$, and $x_t >k $ for $i_0<t<i_n$ with $t\neq$ some $i_r$.  From the Markov property, for $0\le j \le n$, $\mathbb{P}\{x_{i_j+1}=k+1 |  x_{i_j}=k\}=\frac{\rho}{\rho+\lambda k}$ and $\mathbb{P}\{x_{i_j+1}=k-1 | x_{i_j}=k\}=\frac{\lambda k}{\rho+\lambda k}$.  The resulting tree is a Galton-Watson
tree with generation-dependent offspring distributions (see \cite{MR0400434, MR0065835,MR1991122, MR0368197} for more on Galton-Watson trees).  Indeed, we have
$$\mathbb{P}\{\text{a vertex at height $k$  has $n$ offspring}\}=\left(\frac{\rho}{\rho+\lambda k}\right)^n\frac{\lambda k}{\rho+\lambda k},$$ which is the probability of $n$ failures before the first success in a sequence of independent Bernoulli
trials where the probability of success equals $ \frac{\lambda k}{\rho+ \lambda k}$. 

\subsection{Fragment lengths have a general distribution}
Suppose that we have a general distribution $\mu$ for the fragment
lengths.   We observe $X$ at some fixed ``time'' -- which might as well
be $0$ because of stationarity, and ask for the conditional
probability given $X_0$ that the next jump of $X$ will be upwards.
We know from the above that if $\mu$ is exponential with
rate $\lambda$, then conditional on $X_0 = k$ this is 
$\rho/(\rho + k \lambda)$.

Let $T$ denote the time until the next segment comes along.
This random variable has an exponential distribution with rate $\rho$ 
and is independent of $X_0$ \cite[\textsection 2.1]{Daley-book}. 
If we condition on $X_0 = k$, the two-dimensional Poisson point process
must have $k$ points in the region
\[
A:=\{(t,l) : -\infty < t \le 0, \, -t < l < \infty\}.
\]

\begin{figure}[htbp]
\begin{center}
\begin{tikzpicture}
\draw (-4,0) -- (4,0); 
\draw (2,0) -- (-2,4);
\draw (0,0) -- (0,4); 
\draw (0,0) -- (-4,4);
\fill[orange!40!white] (0,0) -- (0,2) -- (-2,4)--(-4,4) -- cycle;
\fill[blue!30!white] (0,2) -- (0,4)--(-2,4) -- cycle;
\draw ( 0cm,2pt) -- (0 cm,-3pt) node[anchor=north] {$0$}; 
\draw ( 2 cm,2pt) -- (2 cm,-3pt) node[anchor=north] {$T$}; 
\draw (3,2.5 )  node[anchor=north] {\small $(t,l)$-plane}; 
\draw (3pt,2 cm) -- (-3pt,2 cm) node[anchor=south west] {$T$}; 

\end{tikzpicture}
\caption{A wedge from the planar Poisson process.}
\label{wedge2}
\end{center}
\end{figure}
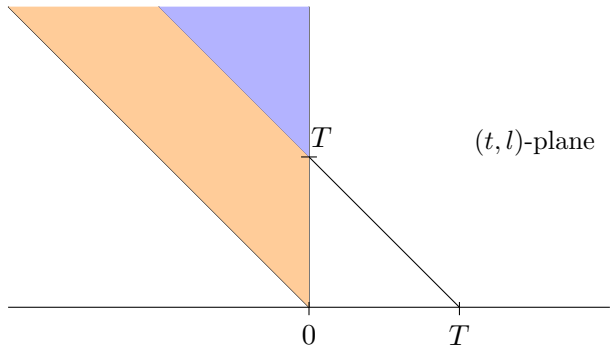
Conditionally, these $k$ points in $A$ have the same distribution as $k$ points chosen at random in $A$ according to the probability measure 
\[
\frac{\rho \, m \otimes \mu(B)}{\rho \, m \otimes \mu(A)} \quad \mbox{for} \quad B\subset A
\]
However, in order that the next jump after
$0$ is upwards, the two-dimensional Poisson point process
must have no points in the orange region
\[
\{(t,l) : -\infty < t \le 0, \, -t < l < T-t\}
\]
as these segments end before time $T$.  This leaves the $k$ points lying in the blue region 
\[
B_T:=\{(t,l) : -\infty < t \le 0, \, T-t \le l < \infty\},
\]  which occurs with probability $\left(
\frac
{
\rho \int_T^\infty \mu((u,\infty)) \, du
}
{
 \rho \int_0^\infty \mu((u,\infty)) \, du
}
\right)^k.$  Thus, conditional on $X_0 = k$,
the probability that the next jump will be upwards is
\[
\int_0^\infty
\left(
\frac
{
 \int_t^\infty \mu((u,\infty)) \, du
}
{
 \int_0^\infty \mu((u,\infty)) \, du
}
\right)^k
\rho e^{-\rho t} \, dt.
\]
Write $p(k)$ for this quantity.  A reasonable approximation to the jump skeleton $Z$ of $X$ is to take it be a discrete-time Markov
chain on the nonnegative integers with transition probabilities
\[
P(i,j) 
=
\begin{cases}
1,& \quad \text{if $i=0$ and $j=1$}, \\
p(i),& \quad \text{if $i \ge 1$ and $j=i+1$}, \\
1 - p(i),& \quad \text{if $i \ge 1$ and $j=i-1$}, \\
0,& \quad \text{otherwise}.
\end{cases}
\]
The resulting tree is then a Galton-Watson tree with generation
dependent offspring distributions, where 
\[
\mathbb{P}\{\text{a vertex at height $k$  has $n$ offspring}\} = p(k)^n (1-p(k)).
\]
\begin{example}
Suppose  $\mu$ is the point mass at $L$
(that is, all segment lengths are $L$).  Then $$\mu((u,\infty))=\begin{cases} 1, & \quad u< L \\ 0,&\quad u\ge L \end{cases},$$ and 
\[
\int_t^{\infty} \mu((u,\infty))du =
\begin{cases}
\int_t^Ldu=L-t, & \quad t<L\\
0, & \quad  t\ge L. 
\end{cases}
\]  This gives 
\begin{eqnarray*}
p(k)&&=\int_{0}^{L}\frac{(L-t)^k}{L^k}\rho e^{-\rho t}dt\\ &&=\int_{0}^{1}w^k\rho e^{-\rho (L-Lw)}Ldw \\ &&= \theta e^{-\theta} \int_0^1w^ke^{\theta w}dw \quad \mbox{for} \quad k\ge 1,
\end{eqnarray*}
 where $\theta := \rho L = \mathbb{E}[X_0]$. 
We integrate by parts and find that $p(k)=\theta e^{-\theta}q(k)$ where 
\begin{eqnarray*}
 q(k) &&= \left.\frac{w^ke^{\theta w}}{\theta}\right|^{w=1}_{w=0} -\frac{k}{\theta}\int_0^1w^{k-1}e^{\theta w}dw \\
 &&= \frac{e^{\theta}}{\theta} -\frac{k}{\theta}q(k-1) \quad \mbox{for} \quad k\ge 2,
\end{eqnarray*}
which yields the recursion 
\[
p(k)=1-\frac{k}{\theta}p(k-1), \quad k\ge 2, \quad \mbox{with} \quad p(1)=1- \frac{1}{\theta}+\frac{e^{-\theta}}{\theta}.
\]
Solving explicitly, we obtain
\[ 
p(k)=k!\left( \sum_{j=0}^k \frac{(-1)^{k-j}}{j!\theta^{k-j}} + \frac{(-1)^{k-1}e^{-\theta}}{\theta^k}\right)\quad \mbox{for} \quad k\ge 1 .
\]
\end{example}

\section{Discussion}

Our observation that randomly sequenced fragments from a genome form a
planar Poisson process in $(position,length)$ coorindates has
implications beyond the coverage function analysis performed in this
paper. For example we have found that the visualization of sequencing data in this
novel form is useful for quickly identifying instances of sequencing bias by eye, as it is
easy to ``see'' deviations from the Poisson process. An example is
shown in Figure \ref{visual} where fragments from an Illumina
sequencing experiment are compared with an idealized simulation (where
the fragments are placed uniformly at random). Specifically, paired-end reads
from an RNA-Seq experiment conducted on a GAII sequencer were mapped
back to the genome and fragments inferred from the read end locations.
Bias in the sequencing is immediately visible, likely due to
non-uniform PCR amplification \cite{Hansen2010} and other effects. We hope that
others will find this approach to visualizing fragment data of use.
\begin{figure}[htbp]
\begin{center}
\includegraphics[width=6in]{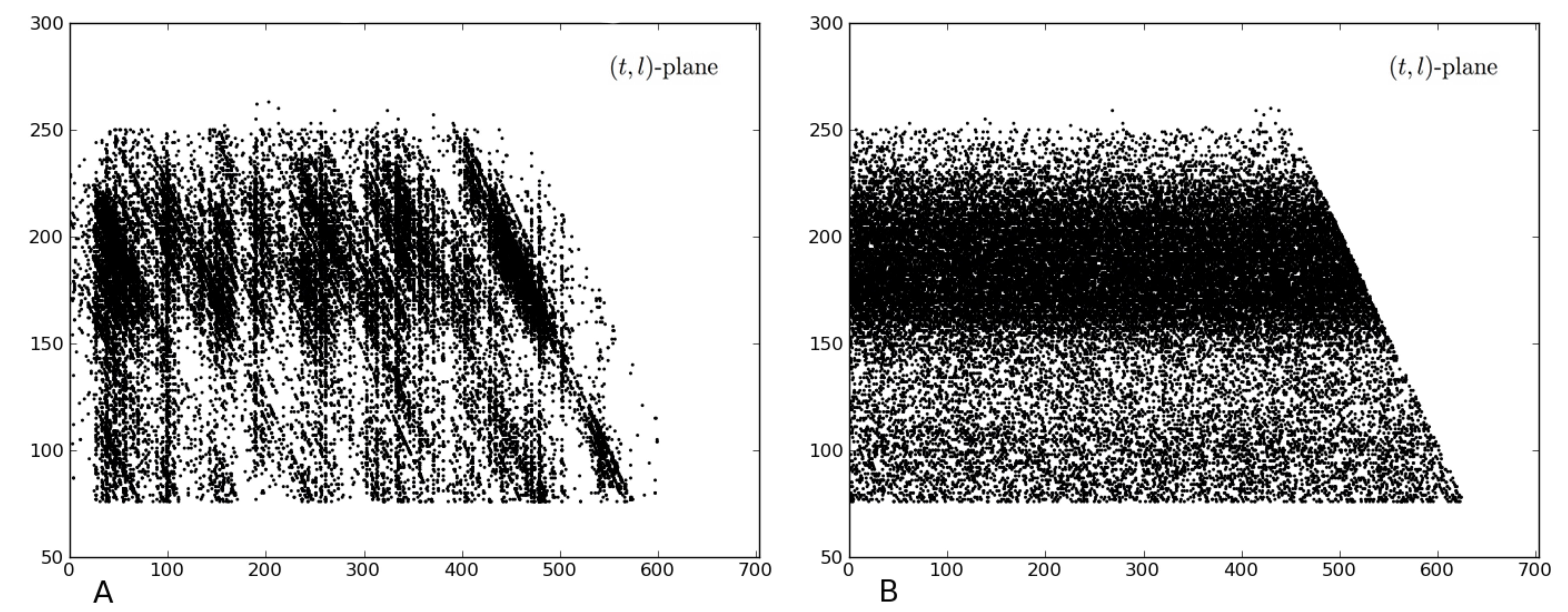}
\label{visual}
\end{center}
\caption{(A) Fragments from a sequencing experiment shown in the ($t,l$) plane. (B) The spatial Poisson process resulting from fragments with the same length distribution as (A) but with position sampled uniformly at random.}
\end{figure}

The ``shape'' we have proposed for coverage functions was motivated by persistence
ideas from topological data analysis (TDA). In the context of TDA, our
setting is very simple (1-dimensional), however unlike what is
typically done in TDA, we have provided a
detailed probabilistic analysis that can be used to construct a null
hypothesis for coverage-based test statistics. For example, we
envision computing test statistics \cite{Matsen2006} based on the trees constructed from
coverage functions and comparing those to the statistics expected from
the Galton-Watson trees. 
 It
should be interesting to perform similar analyses with
high-dimensional generalizations for which we believe many of our ideas
can be translated. There are also biological applications, for example
in the analysis of pooled experiments where fragments may be sequenced
from different genomes simultaneously.

Indeed, we believe that the study of sequence coverage functions that we have initiated
may be of use in the analysis of many sequence census methods. The
number of proposed protocols has exploded in the past two years, as a
result of dramatic drops in the price of sequencing. For example, in
January 2010, the company Illumina announced a new sequencer, the
HiSeq 2000, that they claim ``changes the trajectory of
sequencing'' and can be used to sequence 25Gb per day. Although
technologies such as the HiSeq 2000 were motivated by human genome
sequencing a surprising development has been the fact that the
majority of sequencing is in fact being used for sequence census
experiments \cite{Wold}. The vast amounts of sequence being produced
in the context of complex sequencing protocols, means that a detailed probabilistic
understanding of random sequencing is likely to become increasingly
important in the coming years.
\section{Acknowledgements}SNE is supported in part by NSF grant DMS-0907630 and VH is funded by NSF fellowship DMS-0902723. We thank Adam Roberts for his help in making Figure 4.  
\section{Author Contributions}
LP proposed the problem of understanding the random behaviour of
coverage functions in the context of sequence census methods. VH
investigated the jump skeleton based on ideas from topological data
analysis. SE developed the probability theory and identified the
relevance of Theorem 3.0.1. SNE, VH and LP worked together on all
aspects of the paper and wrote the manuscript.

 \end{document}